\documentclass{birkmult}

\usepackage{amsmath,amssymb,amsfonts,mathrsfs}

\newcommand{\code}[1]{\ensuremath{\mathscr{#1}}}

\newcommand{\wt}[1]{\ensuremath{\textsf{wt}(#1)}}

\newcommand{\prob}[1]{\ensuremath{\textsf{Pr}\left\{#1\right\}}}

\newcommand{\F}{\mathbb{F}}

\newtheorem{proposition}{Proposition}
\newtheorem{lemma}{Lemma}

\newtheorem{remark}{Remark}

\begin{document}

%Version 02.04.2008

 \title[Cryptanalysis of McEliece Cryptosystems Based on Quasi-Cyclic Codes]{Cryptanalysis of  Two McEliece Cryptosystems Based on Quasi-Cyclic Codes}
 \author{Ayoub Otmani}
\address{GREYC - Ensicaen - Universit\'e de Caen,
 Campus II, Boulevard Mar\'echal Juin, F-14050 Caen Cedex, France.}
\email{Ayoub.Otmani@info.unicaen.fr}
\author{Jean-Pierre Tillich}
\address{INRIA, Projet Secret, BP 105, Domaine de Voluceau 
 F-78153 Le Chesnay, France.}
\email{jean-pierre.tillich@inria.fr}
\author{L\'eonard Dallot}
\address{GREYC - Ensicaen - Universit\'e de Caen,
Campus II, Boulevard Mar\'echal Juin, F-14050 Caen Cedex, France.}
\email{Leonard.Dallot@info.unicaen.fr}

\begin{abstract}
  We cryptanalyse here two variants of the McEliece cryptosystem based
  on quasi-cyclic codes. Both aim at reducing the key size by
  restricting the public and secret generator matrices to be in
  quasi-cyclic form.  The first variant considers subcodes of a
  primitive BCH code.  The aforementioned constraint on the public and
  secret keys implies to choose very structured permutations.  We
  prove that this variant is not secure by producing many linear
  equations that the entries of the secret permutation matrix have to
  satisfy by using the fact that the secret code is a subcode of a
  known BCH code. This attack has been implemented and in all
  experiments we have performed the solution space of the linear
  system was of dimension one and revealed the permutation matrix.

  The other variant uses quasi-cyclic low density parity-check codes.
  This scheme was devised to be immune against general attacks working
  for McEliece type cryptosystems based on low density parity-check
  codes by choosing in the McEliece scheme more general one-to-one
  mappings than permutation matrices.  We suggest here a structural
  attack exploiting the quasi-cyclic structure of the code and a
  certain weakness in the choice of the linear transformations that
  hide the generator matrix of the code. This cryptanalysis adopts a
  polynomial-oriented approach and basically consists in searching for
  two polynomials of low weight such that their product is a public
  polynomial.  Our analysis shows that with high probability a
  parity-check matrix of a punctured version of the secret code can be
  recovered with time complexity $O\left( n^3 \right)$ where $n$ is
  the length of the considered code.  The complete reconstruction of
  the secret parity-check matrix of the quasi-cyclic low density
  parity-check codes requires the search of codewords of low weight
  which can be done with about $2^{37}$ operations for the specific
  parameters proposed.

\end{abstract}

\maketitle

\section{Introduction}
Since the introduction of the McEliece public-key cryptosystem \cite{McEliece78}, 
several attempts have been made to propose 
alternatives to the classical Goppa codes. The main  motivation 
is to drastically reduce the size of the public and private keys, which is of real concern 
for any concrete deployment. For instance, the parameters suggested in the original cryptosystem, 
and now outdated, are about 500 Kbits for the public key and 300 Kbits for the private key.
The reason of such a large amount comes from the fact that McEliece proposed 
to use as public key a generator matrix of a linear block code.
He suggested to take a code that admits an efficient decoding algorithm capable 
to correct up to a certain number of errors, and then to hide its structure 
by applying two secret linear transformations: a scrambling transformation 
that sends the chosen generator matrix to another one, and a permutation matrix that reorders 
the coordinates. The resulting matrix is then the public key.
The private key consists in 
the two secret transformations and the decoding algorithm. 

Niederreiter also invented
\cite{Niederreiter86} a code-based asymmetric cryptosystem by choosing to describe 
codes through a parity-check matrix. These two systems are equivalent in terms of security 
\cite{LiDenWan94}.
Their security relies on two difficult problems: the One-Wayness against Chosen-Plaintext
Attack (OW-CPA) thanks to the difficulty of decoding large random linear block codes,
and the difficulty
of guessing the decoding algorithm from a hidden 
generator matrix. It is worthwhile
mentioning that the OW-CPA character is well established as long as appropriate
parameters are taken. This is due to two facts: 
first it is proven in \cite{BMT78} 
that decoding a random linear code  is NP-Hard, and second the best known algorithms 
\cite{CanteautChabaud98,BLP08} and \cite[Volume I, Chapter 7]{Handb98} operate exponentially with the length $n$ 
of the underlying code (see \cite{EOS06} for more details). However, the second criteria is not always verified by any class
of codes that has a decoding algorithm. For instance, Sidel'nikov and Shestakov 
proved in \cite{SidelShesta92} that the structure of Generalised Reed-Solomon codes of length $n$
can be recovered in $O\left(n^3\right)$ (See for instance \cite[page 39]{Gabidulin95}). Sendrier proved \cite{Sendrier98} 
that the permutation transformation can be extracted for concatenated codes.
Minder and Shokrollahi presented in \cite{MinderShokrollahi07} a structural
attack that creates a private key against a cryptosystem based 
on Reed-Muller codes \cite{Sidelnikov94}.

However, despite these attacks on these variants of the
McEliece cryptosystem, the original scheme still remains resistant to any structural attack.
Additionally,  the McEliece system and its Niederreiter homologue display
better encryption and decryption complexity than any other competing 
asymmetric schemes like RSA. 
Unfortunately, they suffer from the same drawback namely, they need very large key sizes
as previously pointed out.
It is therefore crucial to find a method to reduce the   
representation of a linear code as well as the matrices of the linear transformations.

A possible solution is to take  very sparse matrices. This idea has been applied in 
\cite{MonRosShok00} which examined the implications of using 
Low Density Parity-Check (LDPC) codes. The authors showed that taking sparse matrices 
for the linear transformations is not a secure solution. Indeed, it is possible to recover
the secret code from the public parity-check matrix. 
Another idea due to \cite{BergerLoidreau05} is to take subcodes of an optimal code such as Generalized Reed-Solomon codes
in order to decrease the code rate. But a great care has to be taken in the choice of parameters because in \cite{Wieschbrink06} it has been proved
that some parameters are not secure.
A recent trend appeared in code-based public key cryptosystems that tries to 
use quasi-cyclic codes \cite{Gab05,BaldiChiara07,GabLaurSend07,GabGir07,CayOtmVer07}. This particular family of codes offers the advantage of having a very 
simple and compact description. 
Many codewords can simply be obtained by considering cyclic shifts of a sole codeword.
Exploiting this fact leads  to much smaller public and private keys.
Currently there exist two public-key cryptosystems 
based upon quasi-cyclic codes. The first proposal \cite{Gab05} uses subcodes
of a primitive BCH cyclic code. The size of the public key
for this cryptosystem is about 20Kbits.
The other one \cite{BaldiChiara07} tries
to combine these two positive aspects by requiring quasi-cyclic LDPC codes. It also avoids trivial attacks
against McEliece type cryptosystems based on LDPC codes by using in the secret key 
a more general kind of invertible matrix
instead of a permutation matrix. For this particular 
system, the authors propose a public key size that is about 48Kbits.

In this work, we cryptanalyse these two cryptosystems.
We show that the cryptosystem of \cite{Gab05} is not secure because it is possible
to recover the secret permutation that is supposed to hide the structure of the secret quasi-cyclic
code.
We prove  it by producing many linear equations that the entries of the secret permutation matrix have to satisfy by using
the fact that the secret code is a subcode of a known BCH code.
This attack has been implemented and in all experiments we have performed the solution space of the linear
system was of dimension one and revealed the permutation matrix.

In a second part, we also suggest 
 a structural attack of \cite{BaldiChiara07} exploiting the quasi-cyclic structure of the code and
a certain weakness in the choice 
of the linear transformations that hide the generator 
matrix of the code. This cryptanalysis adopts 
a polynomial-oriented  approach and basically
consists in searching for two polynomials of low weight such 
that their product is a public polynomial. 
Our analysis shows that with high probability 
a parity-check matrix of a punctured version of the secret code can be recovered with
time complexity $O\left( n^3 \right)$ 
where $n$ is the length of the considered code. An implementation shows that this recovery
can be done in about $140$ seconds on a PC. The final step that consists in completely 
reconstructing the original parity-check matrix of the secret quasi-cyclic low density parity-check
code requires the search for low weight codewords  which can be done with about 
$2^{37}$ operations for the specific parameters proposed. 

The rest of this paper is organised as follows. In Section~\ref{sec:notation}, 
we recall definitions and basic properties of circulant matrices. 
Section~\ref{sec:BCH} gives a description of how to totally break the McEliece variant 
proposed in \cite{Gab05}. In Section~\ref{sec:ldpc} we propose a method to totally 
cryptanalyse the scheme of \cite{BaldiChiara07}. Section~\ref{sec:conclusion} 
concludes the paper.

\section{Notation and Definitions} \label{sec:notation}
\subsection{Circulant Matrices}

Let $\F_2$ be the finite field with two elements and denote by $\F_2[x]$ 
the set of univariate polynomials with coefficients in $\F_2$. Any $p$-bit
vector $\boldsymbol{v}= (v_0,\dots{},v_{p-1})$ is identified to the polynomial 
$\boldsymbol{v}(x) = v_0 + \cdots{} v_{p-1}x^{p-1}$.
The \textit{support} of a vector (or a polynomial) $\boldsymbol{v}$ 
is the set of positions $i$ such that $v_i$ is non-zero and the \textit{weight} $\wt{\boldsymbol{v}}$ of $\boldsymbol{v}$ is
the cardinality of its support. 
The \textit{intersection} polynomial  for any two polynomials 
$\boldsymbol{u}(x)$ and $\boldsymbol{v}(x)$ is 
$\boldsymbol{u}(x) \star \boldsymbol{v}(x) = \sum u_i v_i x^i$. 

A binary \textit{circulant} matrix $M$ is a $p \times p$ matrix 
obtained by cyclically right shifting the first row:
\begin{equation}
M = \left ( \begin{array}{*{4}{l}}
m_0        & m_1   & \cdots{} & m_{p-1}\\
m_{p-1}   & m_0      &   \cdots{} & m_{p-2} \\
\vdots{}   & \vdots{}      &  \ddots{} & \vdots{} \\
m_1    & m_2   &   \cdots{} & m_0
\end{array}
\right).
\end{equation}
Thus any circulant matrix $M$ is completely described by only its first row
$\boldsymbol{m} = (m_0,\dots{},m_{p-1})$.
Note that a circulant matrix is also obtained by cyclically down shifting its first column. 
We shall  see that the classical matrix operations of addition and multiplication preserve  the 
circulant structure of matrices. 
It is possible to characterise 
the $i$-th row  of a circulant matrix $M$ as the polynomial:
$$
x^i \cdot{} \boldsymbol{m}(x) \mod (x^p-1).
$$
If one looks at the product $\boldsymbol{b} \times M$ 
of a circulant matrix $M$ with a binary vector 
$\boldsymbol{b} = (b_0,\dots{},b_{p-1})$ then it exactly corresponds to the $p$-bit vector
represented by the polynomial $\boldsymbol{b}(x)  \cdot{} \boldsymbol{m}(x)  \mod (x^p-1)$.
This property naturally extends to the product of two $p \times p$ circulant matrices 
$M$ and $N$. Indeed, the first row of $M \times N$ is exactly 
$\boldsymbol{m}(x) \cdot{}  \boldsymbol{n}(x) \mod (x^p-1)$
and the $i$-th row of $M \times N$ is represented 
by the polynomial:  
$$
\Big(x^i\cdot{}\boldsymbol{m}(x)\Big) \cdot{}  \boldsymbol{n}(x) \mod (x^p-1) = 
x^i \cdot{} \Big(\boldsymbol{m}(x) \cdot{}\boldsymbol{n}(x)\Big) \mod (x^p-1).
$$
We have therefore the following result.
\begin{proposition}
\label{prop:isomorphism}
Let $\mathfrak{C}_p$ be the set of binary $p\times p$ circulant matrices,
then  there 
exists an isomorphism  between the rings 
$\Big(\mathfrak{C}_p,+,\times\Big)$ and $\Big(\F_2[x]/(x^p-1),+,\cdot{}\Big)$: 
$$
\Big(\mathfrak{C}_p,+,\times \Big) \simeq \Big(\F_2[x]/(x^p-1),+,\cdot{} \Big) 
$$
\end{proposition}

\begin{remark}
The first column of a circulant matrix $M$ defined by $\boldsymbol{m}(x)$ corresponds to  the polynomial
$\boldsymbol{m}^{\star}(x) = x^{p}\cdot{}\boldsymbol{m}(\frac{1}{x}) \mod (x^p-1)$. 
\end{remark}
Proposition \ref{prop:isomorphism} can be used to provide
a simple characterisation
of invertible matrices of circulant matrices:
\begin{proposition}
A $p \times p$ circulant matrix $M$ is invertible if and only if $\boldsymbol{m}(x)$ is prime with $x^p-1$.
\end{proposition}

\begin{proof}
One has only to prove that the invert of a circulant matrix $M$ defined by a polynomial  $\boldsymbol{m}(x)$ of $\F_2[x]/(x^p-1)$
is necessarily a circulant matrix. Assume that there exists $N$ such that $N \times M = M \times N = I_p$ with $I_p$ being the
$p \times p$ identity matrix. Let $\boldsymbol{n} = (n_0,\dots{},n_{p-1})$ be the first row of $N$. We have previously seen that
the product $\boldsymbol{n}\times M$ can be seen as the polynomial 
$\boldsymbol{n}(x) \cdot{} \boldsymbol{m}(x) \mod (x^p-1)$. This latter polynomial is equal to $1$ by assumption. 
Consequently, for any $i$ such that $0 \le i \le p-1$ we also have 
$ \Big(x^i\cdot{} \boldsymbol{n}(x) \Big)\cdot{} \boldsymbol{m}(x) = x^i \mod (x^p-1)$ which proves that the circulant matrix defined
by $\boldsymbol{n}(x)$ is the invert of $M$. Therefore $N$ is circulant.
\end{proof}

A matrix $G$ of size $k \times n$  is $p$-\textit{block circulant} 
with $k = k_0 p $ and $n = n_0 p$ 
where $k_0$ and $n_0$ are  positive integers if there exist $p \times p$ circulant matrices
$G_{i,j} \in \mathfrak{C}_p$ such that:
$$
G = \left (
\begin{array}{*{3}{c}}
G_{1,1}   &  \cdots{} & G_{1,n_0}\\
\vdots{} &    & \vdots{} \\
G_{k_0,1} & \cdots{}   & G_{k_0,n_0}
\end{array}
\right)
$$
It is straightforward to see that the set of block circulant matrices is stable by matrix 
addition and matrix multiplication. It is
therefore natural  to establish an identification between a block circulant matrix $G$ with a 
polynomial $k_0 \times n_0$ matrix 
$\boldsymbol{G}(x)$ with entries in $\F_2[x]/(x^p-1)$ by means of the mapping that sends 
each block $G_{i,j}$ to the polynomial $\boldsymbol{g}_{i,j}(x)$ defining it.

\begin{proposition}
Let  $\mathfrak{B}_{k_0,n_0}^p$ be the set of $p$-block circulant
matrices of size $k_0\times n_0$. 
Let $R_p = \F_2[x]/(x^p-1)$ and define by $\mathfrak{M}_{k_0,n_0}(R_p)$ 
the set of $k_0\times n_0$ 
matrices  with coefficients in $R_p$.
There exists a ring isomorphism 
between $\mathfrak{B}_{k_0,n_0}^p$ and $\mathfrak{M}_{k_0,n_0}(R_p)$:
$$
\begin{array}{rcl}
\mathfrak{B}_{k_0,n_0}^p &\simeq& \mathfrak{M}_{k_0,n_0}(R_p)\\
G   & \longmapsto & \boldsymbol{G}(x).
\end{array}
$$
\end{proposition}
In particular any  $p$-block circulant 
matrix $G$ is invertible if and only if $\det(\boldsymbol{G})(x)$ is prime with $x^p-1$
and its inverse  is also a $p$-block circulant matrix.

\subsection{Cyclic  and Quasi-Cyclic Codes} \label{sec:CyclicDef}
A (binary) linear code $\code{C}$ of length $n$ and dimension $k$ is a $k$-dimensional 
vector subspace of $\F_2^n$. The elements of a code are called \textit{codewords}. 
A \textit{generator matrix} $G$ of $\code{C}$ is a $k'\times n$ 
matrix with $k' \ge k$ whose rows generate $\code{C}$. A \textit{parity-check} matrix $H$ of $\code{C}$ is
an $r \times n$ matrix with $r \ge n-k$ such that for any codeword $\boldsymbol{c} \in \code{C}$ we have:
$$
H \times \boldsymbol{c}^T = 0.
$$
It is well-known that if a generator matrix of $\code{C}$ is of the form $(I | A)$
where $I$ is the identity matrix then $(A^T | I)$ is a parity-check matrix for $\code{C}$. Such a generator matrix is said
to be \textit{in reduced echelon form}. A code $\code{C}'$ is said to be \textit{permutation equivalent} to
$\code{C}$ if there exists a permutation 
of the symmetric group of order $n$ that reorders the coordinates of 
codewords of $\code{C}'$ into codewords of $\code{C}$. It is convenient to consider equivalent codes
as the same code.

\medskip

A \textit{cyclic} code $\code{C}$ of length $n$ is an ideal of the ring $\F_2[x]/(x^n-1)$. Such a code is characterised by
a unique polynomial $\boldsymbol{g}(x)$ divisor of $(x^n-1)$. Let $r$ be the degree of $\boldsymbol{g}(x)$. 
Any codeword $\boldsymbol{c}(x)$ is obtained as a product in 
$\F_2[x]$ of the form:
$$
\boldsymbol{c}(x) = \boldsymbol{m}(x) \cdot{} \boldsymbol{g}(x)
$$
where $\boldsymbol{m}(x)$ is a polynomial of $\F_2[x]$ of degree $n-1-r$. $\code{C}$ is a linear code of dimension 
$k = n - r$. The polynomial $\boldsymbol{g}(x)$ is called the \textit{generator polynomial} of the cyclic code $\code{C}$
and we shall write $\code{C} = < \boldsymbol{g}(x)>$.

A code $\code{C}$ is \textit{quasi-cyclic of index} $p$ if there exists a 
generator matrix $G$ that is $p$-block circulant. We assume that all the $G_{i,j}$'s are square matrices 
of size $p \times p$ and therefore $n = n_0 p$ and $k = k_0 p$. Cyclic codes of length $n$ 
are thus quasi-cyclic codes of index $n$ where a generator matrix is a circulant matrix 
associated to its generator polynomial.

\medskip

A useful method developed in \cite{Gab05}
for obtaining quasi-cyclic codes of length $n=pn_0$ and index $p$ is to consider 
a cyclic code $\code{C}$ generated by a polynomial $\boldsymbol{g}(x)$ 
and construct the subcode $S_{n_0}(\boldsymbol{c})$ spanned by a codeword $\boldsymbol{c}(x)$ and its $p-1$ shifts 
modulo $(x^n - 1)$ of $n_0$ bits $x^{n_0}\cdot{} \boldsymbol{c}(x),\dots{},x^{(p-1)n_0}\cdot{} \boldsymbol{c}(x)$. 
However  note that $S_{n_0}(\boldsymbol{c})$ does not admit a $p$-block circulant generator matrix. Actually, one has to
consider the equivalent code of $\code{C}$ obtained with the permutation $\pi$ that maps any 
$a n_0 + b$ to $b p + a$ with $1 \le a \le p-1$ and $0\le b \le n_0-1$. It means that up 
to a permutation any codeword $\boldsymbol{c}(x)$ of a cyclic code $\code{C}$ 
can be seen as a vector $\boldsymbol{c} = (\boldsymbol{c}_0,\dots{},\boldsymbol{c}_{n_0-1})$ where each 
$\boldsymbol{c}_i$ belongs to $\F_2^p \simeq \F_2[x]/(x^p-1)$ and such that the vector 
$\boldsymbol{c}' = (\boldsymbol{c}'_0,\dots{},\boldsymbol{c}'_{n_0-1})$ with 
$\boldsymbol{c}'_{j}(x) = x \cdot{} \boldsymbol{c}_{j}(x) \mod (x^p-1)$ 
is also a codeword of $S_{n_0}(\boldsymbol{c})$. 

\section{A McEliece Cryptosystem Based on Subcodes of a BCH Code} \label{sec:BCH}
\subsection{Description}

Let $\code{C}_0$ be a cyclic code of length $n = p n_0$ and
let $k$ be the dimension of $\code{C}_0$. 
Assume that  $\code{C}_0$ admits an $k' \times n$
generator matrix with $k' \ge k$ and such that $k' = p k_0$. For simplicity, 
we set $k' = k$.
Let $\boldsymbol{c}_1(x)$,
$\boldsymbol{c}_2(x)$,\dots{},$\boldsymbol{c}_{k_0-1}(x)$ be random codewords of $\code{C}_0$ and
consider the linear code $\code{C}$ defined as:
$$
\code{C} = S_{n_0}(\boldsymbol{c}_1) + \cdots{} + S_{n_0}(\boldsymbol{c}_{k_0-1}).
$$
We assume that $\code{C}$ is of dimension $k-p = p (k_0-1)$. Recall from
Section~\ref{sec:CyclicDef} that up 
to a permutation any $n$-bit vector $\boldsymbol{c}_i(x)$ with $1\le i\le k_0-1$ can be seen as 
a vector $(\boldsymbol{c}_{i,0},\dots{},\boldsymbol{c}_{i,n_0-1})$ where each 
$\boldsymbol{c}_{i,j}$ can also be seen as an element of $\F_2[x]/(x^p-1)$. Thus 
$\code{C}$ is a quasi-cyclic code of index $p$ whose generator matrix $\boldsymbol{G}(x)$ 
in $p$-block circulant form is: 
$$
\boldsymbol{G}(x) = 
\left (
\begin{array}{*{3}{c}}
\boldsymbol{c}_{1,1}(x)   &  \cdots{} & \boldsymbol{c}_{1,n_0}(x)\\
\vdots{} &    & \vdots{} \\
\boldsymbol{c}_{k_0-1,1}(x) & \cdots{}   & \boldsymbol{c}_{k_0-1,n_0}(x)
\end{array}
\right).
$$

The variant of the McEliece cryptosystem proposed in \cite{Gab05} starts from a secret subcode $\code{C}$ 
of dimension $ p (k_0-1)$ of a primitive BCH code $\code{C}_0$ obtained by the method explained above. 
A secret permutation  $\pi$ of the symmetric group of order $n_0$ hides the structure 
of $\code{C}$ while keeping its quasi-cyclic structure by publicly making  available a
generator matrix $\boldsymbol{G}^{\pi}(x)$ defined by:
$$
\boldsymbol{G}^{\pi}(x) = 
\left (
\begin{array}{*{3}{c}}
\boldsymbol{c}_{1,\pi(1)}(x)   &  \cdots{} & \boldsymbol{c}_{1,\pi(n_0)}(x)\\
\vdots{} &    & \vdots{} \\
\boldsymbol{c}_{k_0-1,\pi(1)}(x) & \cdots{}   & \boldsymbol{c}_{k_0-1,\pi(n_0)}(x)
\end{array}
\right).
$$
The cyclic code $\code{C}_0$ given in \cite{Gab05} is  a primitive BCH of length $2^m-1$ and dimension $n-tm$ where
$t$ is a positive integer. Two sets of parameters
are proposed respectively 
corresponding to $2^{100}$ and $2^{80}$ security levels.
\begin{itemize}
\item Parameters A: $m=12$, 
$t=26$, $p=91$, $n_0 = 45$, and $k_0 = 43$.
\item Parameters B:  $m=11$, $t=31$, $p=89$, $n_0 = 23$ and $k_0 = 21$.
\end{itemize}
Note that we always have $p > n_0$. This property will be useful for cryptanalyzing the cryptosystem.

\subsection{Structural Cryptanalysis}
We describe a method that recovers the secret permutation $\pi$ of 
the cryptosystem of \cite{Gab05} and thus reveals the secret 
key of any user. 
It exploits three facts: 
\begin{enumerate}
\item The code $\code{C}_0$ admits a binary $(n-k) \times n$ 
parity check matrix $H_0$ which can be assumed to be known. 
There are only a few different primitive BCH codes for a given parameter set $(n,m,t)$ and we can try all of them.
This is a consequence of the fact that the number of such codes is clearly upper-bounded
by the number of primitive polynomials of degree $m$.
For instance for the parameter set B, this number is equal to $176$.
 
\item Since $\code{C}$ is a subcode of $\code{C}_0$, any $n$-bit codeword 
$\boldsymbol{c}$ of $\code{C}$ must satisfy the equation:
\begin{equation} \label{BCHkeyEquation}
H_0 \times \boldsymbol{c}^T = 0.
\end{equation}

\item 
Permuting through a permutation $\pi$ the columns of a polynomial generator matrix 
$\boldsymbol{G}(x)$ 
of  $\code{C}$ can also be translated into a matrix product 
by the associated $n_0 \times n_0$ 
permutation matrix $\boldsymbol{\Pi}$ of $\pi$. Note that $\boldsymbol{\Pi}$ can also be seen as a polynomial matrix $\boldsymbol{\Pi}(x) \in \mathfrak{B}_{n_0,n_0}^p$ 
where 0 (resp. 1) entry corresponds to 0 (resp. 1) constant polynomial so that we have:
\begin{equation} \label{Eq:HidingEquation}
\boldsymbol{G}^{\pi}(x) =  \boldsymbol{G}(x) \times \boldsymbol{\Pi}(x). 
\end{equation}
Note that Equation~(\ref{Eq:HidingEquation}) can be rewritten as an equality between 
binary $p$-block circulant matrices:
\begin{equation} \label{Eq:BinaryHidingEquation}
G^{\pi} = G \times \Pi,
\end{equation}
where $G^{\pi}$ is the $(k-p) \times n$ public generator matrix and $\Pi = \boldsymbol{\Pi} \otimes I_p$
with $I_p$ being the $p \times p$ identity matrix. 
 Finding  $\boldsymbol{\Pi}$ actually amounts to solve a linear system of 
$n_0^2$ unknowns representing the entries of $\boldsymbol{\Pi}^{-1}$ such that:
\begin{equation} \label{Eq:UnknownPerm}
H_0 \times  \left(G^{\pi} \times \Pi^{-1} \right)^T = 0.
\end{equation}
In other words, each row of the public matrix $G^{\pi}$ after being permuted by $\Pi^{-1}$ must satisfy 
Equation~(\ref{BCHkeyEquation}). This is a linear system since $\Pi^{-1}$ may be rewritten as $\boldsymbol{\Pi}^{-1}\otimes I_p$.
 This means that each row of $G^{\pi}$ provides
$(n-k)$ binary linear equations verified by  $\Pi^{-1}$. Thus Equation~(\ref{Eq:UnknownPerm}) gives 
a total number of $(k-p)(n-k)$ linear equations that must be satisfied by $n_0^2$ unknowns. 
\end{enumerate}
The cryptanalysis of \cite{Gab05} amounts to solve an over-constrained
linear system constituted of $p^2(k_0-1)(n_0-k_0)$ equations and $n_0^2$ unknowns since 
as we have remarked that $p > n_0$. For instance, Parameters~B  give $529$ unknowns that should
satisfy $316,840$ equations.
As for Parameters~A we obtain $2,025$ unknowns that satisfy $695,604$
equations. Many of these equations are obviously linearly dependent.
The success of this method heavily depends 
on the size of the solution vector space.  
An implementation in Magma software actually
always gave in both cases a vector space of dimension one.
This revealed the secret permutation.

\section{A Cryptosystem Based on Quasi-Cyclic LDPC Codes} \label{sec:ldpc}

\subsection{Description}

LDPC codes are linear codes defined by  sparse binary 
parity-check matrices. We assume as in \cite{BaldiChiara07}
that $n = p n_0$ and $k = p (n_0-1)$, and we consider a 
parity-check matrix $H$ of the following form:
\begin{equation} \label{eq:PC-LDPC}
H = \left (
\begin{array}{*{3}{c}}
H_1 & \cdots{} & H_{n_0}
\end{array}
\right)
\end{equation}
where each matrix $H_j$ is a sparse circulant matrix of size $p \times p$. Without loss of generality,
$H_{n_0}$ is chosen to have full rank.
Each column of $H$ has a fixed weight $d_v$ which is very small compared to the length $n$. We also assume that
one has a good approximation of the number $t$ of correctable
errors through iterative decoding of the code defined
by $H$.

The quasi-cyclic LDPC cryptosystem proposed in \cite{BaldiChiara07} takes two
invertible $p$-block circulant matrices $S$ and $Q$ of size 
$k \times k$ and $n \times n$ respectively. The matrix $S$ (resp. $Q$) is chosen such that the weight of
each row and each column is $s$ (resp. $m$).
The private key consists of the parity-check matrix $H$ and  the matrices $S$ and $Q$.
In order to produce the public key, one has to compute  a generator matrix $G'$ in reduced echelon 
form and make public the matrix $G = S^{-1} \times G' \times Q^{-1}$. The plaintext space 
is the set $\F_2^k$ and the ciphertext space is $\F_2^n$. If one wishes to encrypt a message 
$\boldsymbol{x} \in \F_2^k$,  one has to randomly choose a $n$-bit vector $\boldsymbol{e}$ 
of weight $t' \le t/m$ and compute
$\boldsymbol{c} = \boldsymbol{x} \times G + \boldsymbol{e}$. 
The decryption step consists  in iteratively decoding 
$\boldsymbol{c} \times Q = \boldsymbol{x} \times S^{-1} \times G' + \boldsymbol{e}\times Q$ to output $\boldsymbol{z}=\boldsymbol{x} \times S^{-1}$ and then computing $\boldsymbol{x} =  \boldsymbol{z} \times S$. 
The crucial point that makes this cryptosystem valid is that  $\boldsymbol{e} \times Q$ is  
a  correctable error because its weight is less than or equal to $t'm$.

\subsection{Some Remarks on the Choice of the Parameters}

The authors suggest to take a matrix
$Q$ in diagonal form. They also suggest
the following values: 
$p = 4032$, $n_0 = 4$, $d_v = 13$ , $m=7$ and $t = 190$ ($t'= 27$). 
Finally, each block circulant matrix of 
$S$ has a column/row weight equals to $m$ so as to  have $s = m(n_0-1)$. 
Unfortunately, for this specific constraint, there is a flaw in this choice because the matrix $S$  
is not invertible. This follows from the fact that in this case $x-1$ always divides 
$\det(\boldsymbol{S})(x)$ which is therefore not coprime with $x^p-1$ and this implies that $\boldsymbol{S}(x)$ is not invertible. This can be proved by using the following arguments.

\begin{lemma} \label{lemma:weight}
Let $\boldsymbol{S}(x) = (\boldsymbol{s}_{i,j}(x))$  in $\mathfrak{M}_{n_0-1,n_0-1}(R_p)$ and define
the binary matrix $\tilde{S} = (\tilde{s}_{i, j})$ by 
$\tilde{s}_{i, j} = \wt{\boldsymbol{s}_{i, j}} \mod 2$. We have 
then:
$$
\normalfont
\det(\tilde{S}) = \wt{\det(\boldsymbol{S})} \mod 2.
$$
\end{lemma}
\begin{proof}
This comes from the fact that 
$\wt{\boldsymbol{u}+\boldsymbol{v}} = \wt{\boldsymbol{u}} + \wt{\boldsymbol{v}} -
2\wt{\boldsymbol{u}\star \boldsymbol{v}}$ for any $\boldsymbol{u}(x)$ and $\boldsymbol{v}(x)$ 
in $\F_2[x]$ which implies that:
$$
\left\{
\begin{array}{lcl}
\wt{\boldsymbol{u}+\boldsymbol{v}} &=& \wt{\boldsymbol{u}} + \wt{\boldsymbol{v}} \mod 2\\ 
\wt{\boldsymbol{u} \cdot{} \boldsymbol{v}} &=& \wt{\boldsymbol{u}} \cdot{} \wt{\boldsymbol{v}} \mod 2.
\end{array}
\right.
$$
\end{proof}

\begin{proposition} For any $\boldsymbol{S}(x)$ in  
$\mathfrak{M}_{3,3}(R_p)$ such that each $\boldsymbol{s}_{i,j}$ 
is of weight $m$ then $x-1$ divides $\det(\boldsymbol{S})(x)$.
\end{proposition}
\begin{proof}
By using the same notation as in the previous lemma we know that $\det(\tilde{S})$ is equal to zero since $\tilde{S}$ is the all one matrix.
From the previous lemma it follows that
$\det(\boldsymbol{S})(x)$ has a support of even weight. This implies that $x-1$ divides $\det(\boldsymbol{S})(x)$.
\end{proof}

In order to avoid this situation we introduce as few  polynomials of weight different from $m$ in $\boldsymbol{S}$
such that $\det(\tilde{S}) = 1$. A possible choice is the following one. First we choose a nonsingular $\tilde{S}$  equal to
$$
\tilde{S} = \left(
  \begin{array}{ccc}
    1 & 1 & 1\\
    1 & 0 & 1\\
    0 & 1 & 1\\
  \end{array}
\right)
$$
When $\tilde{s}_{ij}=1$ we choose the corresponding entry $\boldsymbol{s}_{ij}(x)$ to be of weight $m$ and if
$\tilde{s}_{ij}=0$ we choose the corresponding entry $\boldsymbol{s}_{ij}(x)$ to be of weight $m-1$.

It should also be mentioned that a decoding attack searching for a
 word of weight less than $t = 27$ in a code of length $n=
 16128$ and dimension $k=12096$ as proposed by  using
 the algorithm given in \cite{CC94} has a work factor of about
 $2^{78.5}$. Note that this work
 factor may even be decreased with  the algorithm of  \cite{CC95}.

\subsection{Structural Attack}
\subsubsection{Preliminaries}

The goal of this attack is to recover the secret code $\code{C}$ defined by the parity-check
matrix $H$ given in Equation~(\ref{eq:PC-LDPC}).
We know that $S$ and $Q$ are equivalently defined by polynomials $\boldsymbol{s}_{i,j}(x)$ and 
$\boldsymbol{q}_{i,j}(x)$ respectively. $Q$ is chosen to be in diagonal form, that is to say 
$\boldsymbol{q}_{i,j}(x) = 0$ if $i \not = j$. For the sake of simplicity, 
we set $\boldsymbol{q}_i(x) = \boldsymbol{q}_{i,i}(x)$. Moreover 
the polynomials $\boldsymbol{q}_i(x)$ are  invertible modulo $x^p-1$ since $Q$ is invertible. 
It is also straightforward
to remark that the secret generator matrix $G'$ is equal to:
$$
G' = \left (
\begin{array}{ccc|c}
            & & & (H_{n_0}^{-1} H_1)^T\\
& I_{k} &   & \vdots{}\\
&       &     & (H_{n_0}^{-1} H_{n_0-1})^T
\end{array}
\right).
$$ 
In others words, if we denote by $G_{\le k}$ the matrix obtained by taking the 
$k$ first columns of $G$ then we have:
$$
G_{\le k} = S^{-1}
\times 
\left (
\begin{array}{*{4}{c}}
Q_1^{-1}   & 0 &  \cdots{} & 0\\
0  & \ddots{}  & \ddots{} & \vdots{} \\
\vdots{} & \ddots{}  & \ddots{} & 0 \\
0 & \cdots{}   & 0 & Q_{n_0-1}^{-1}
\end{array}
\right).
$$
This implies that $G_{\le k}^{-1}$ is a $p$-block circulant matrix defined by polynomials 
$\boldsymbol{g}_{i,j}(x)$ that satisfies the following equations:
\begin{equation} \label{keyEq:LDPC}
\boldsymbol{g}_{i,j}(x) = \boldsymbol{q}_{i}(x) \cdot{} \boldsymbol{s}_{i,j}(x)  \mod (x^p-1).
\end{equation}
Note that the weight of $\boldsymbol{g}_{i,j}(x)$ is at most $m^2$.  Actually,
due the fact that the secret 
polynomials have very low weights, 
we shall see that 
the support of $\boldsymbol{g}_{i,j}(x)$ is exactly $m^2$ with a good probability. 
For the sake of simplicity, we set
$\boldsymbol{q}_{i}(x) = x^{e_1}+ \cdots{} + x^{e_m}$ and 
$\boldsymbol{s}_{i,j}(x) = x^{\ell_1}+ \cdots{} + x^{\ell_m}$ with $0\le e_a \le p-1$ 
and $0 \le \ell_a \le p-1$ for any $1 \le a \le m$.
We fix $\boldsymbol{q}_{i}(x)$ and we assume that
the monomials $x^{\ell_a}$ of $\boldsymbol{s}_{i,j}(x)$ are 
independently and uniformly chosen. We wish to estimate 
the probability that the support of $\boldsymbol{g}_{i,j}(x)$ contains 
the support of at least  one shift $x^{\ell_a} \cdot{} \boldsymbol{q}_{i}(x)$,
and 
the probability that the weight of $\boldsymbol{g}_{i,j}(x)$ is exactly $m^2$. 

\begin{lemma} \label{lem:DisjointShifts}
Let $\ell_1,\dots{},\ell_w$ be $w$ different integers such that 
$0\le \ell_a \le p-1$  for $1 \le a \le w$. For any  random integer
$0\le \ell \le p-1$ such that $\ell$ is different from
$\ell_1,\dots{},\ell_w$, we have:  
$$
\prob{\left(x^{\ell_1}+\cdots{} + x^{\ell_w} \right) \cdot{} \boldsymbol{q}_{i}(x) \star x^{\ell} \cdot{} \boldsymbol{q}_{i}(x)
\not = 0} \le w\frac{m(m-1)}{p-w}
$$
\end{lemma}

\begin{proof}
Set first $\boldsymbol{r}(x) = \left(x^{\ell_1}+\cdots{} + x^{\ell_w} \right) \cdot{} \boldsymbol{q}_{i}(x)$.
By the union bound we have: 
$$
\prob{\boldsymbol{r}(x) \star x^{\ell} \cdot{} \boldsymbol{q}_{i}(x) \not = 0} \le
\sum_{a=1}^w \prob{x^{\ell_a} \cdot{} \boldsymbol{q}_{i}(x) \star x^{\ell} \cdot{} \boldsymbol{q}_{i}(x) \not = 0}
$$
The probability $\prob{x^{\ell_a} \cdot{} \boldsymbol{q}_{i}(x) \star x^{\ell} \cdot{} \boldsymbol{q}_{i}(x) \not = 0}$
is at most the fraction of integers $\ell$ 
different from $\ell_1,\dots{},\ell_w$ such that  there exist 
$1 \le b \le m$ and $1 \le c \le m$ with:
$$
\ell_a +e_b = \ell+ e_c \mod p.
$$ 
Thus, this fraction is given by the ratio 
of the number of pairs $(e_b,e_c)$ with $b \not = c$ to the number of possible 
values for $\ell$ which is exactly $m(m-1)/(p-w)$.
\end{proof}

\begin{proposition}
The probability 
$\prob{x^{\ell} \cdot{} \boldsymbol{q}_{i}(x) \subset \boldsymbol{g}_{i,j}(x)}$
for $\ell$ in $\{\ell_1,\dots{},\ell_m\}$ 
that the support of $\boldsymbol{g}_{i,j}(x)$
contains the support of 
$x^\ell\cdot{} \boldsymbol{q}_{i}(x)$ is lower-bounded by:
$$
\prob{x^{\ell} \cdot{} \boldsymbol{q}_{i}(x) \subset \boldsymbol{g}_{i,j}(x)}
 \ge \left(1- \frac{m(m-1)}{p-1}\right)^{m-1}.
$$

\end{proposition}

\begin{proof}
This inequality is obtained by taking $w=1$ in Lemma~\ref{lem:DisjointShifts}
and by the independence of the choice of the $(m-1)$
other monomials of $\boldsymbol{s}_{i,j}(x)$.
\end{proof}

\begin{proposition}
The probability $q$  that $\boldsymbol{g}_{i,j}(x)$ is exactly of weight 
$m^2$ is lower-bounded by:
$$
q \ge \prod_{w=1}^{m-1} \left( 1 - w\cdot{} \frac{m(m-1)}{p-w} \right).
$$
\end{proposition}

\begin{proof}
For any $2 \le w \le m$, let $E_w$ denote the event that  
$$
E_w ~:~(x^{\ell_1}+ \cdots{} + x^{\ell_{w-1}})\cdot{}\boldsymbol{q}_{i}(x)\star  x^{\ell_w}\cdot{}\boldsymbol{q}_{i}(x) = 0
$$
when each monomial 
$x^{\ell_a}$ is uniformly and independently chosen. We also set $E_1$ as the whole universe.
Then we have:
$$
q \ge \prob{E_2\cap\cdots{}\cap E_m}
$$
Using Bayes' rule we also have  
$$
\prob{E_2\cap\cdots{}\cap E_m} = \prod_{w=1}^m\prob{E_w | E_{w-1}\cap\cdots{}\cap E_1}.
$$
But by 
%the 
Lemma~\ref{lem:DisjointShifts} we know that 
$\prob{E_w | E_{w-1}\cap\cdots{}\cap E_1} \ge \left( 1 - w\cdot{} \frac{m(m-1)}{p-w} \right)$.
\end{proof}
\subsubsection{Different Strategies}
\paragraph{First Strategy}

We have seen in Lemma~\ref{lem:DisjointShifts} that the support of 
$\boldsymbol{g}_{i,j}(x)$ contains with very high probability the support of
at least\footnote{Actually, the support of 
$\boldsymbol{g}_{i,j}(x)$ contains with  good probability 
all the supports of $x^{\ell_a}\cdot{} \boldsymbol{q}_{i}(x)$ with $1\le a \le m$
since $q \ge 0.79$ for the proposed parameters.}
a shifted version of $\boldsymbol{q}_{i}(x)$ since
for the parameters given in \cite{BaldiChiara07}, 
we obtain $\prob{x^{\ell} \cdot{} \boldsymbol{q}_{i}(x) \subset \boldsymbol{g}_{i,j}(x)} \ge 0.94$. 
One possible strategy to recover the polynomial $\boldsymbol{q}_{i}(x)$ consists in enumerating $m$-tuples $u_1,\dots{},u_m$ 
that belong in the support of $\boldsymbol{g}_{i,j}(x)$ in order to form
$\boldsymbol{u}(x) =\sum_{a} x^{u_a}$ such that $\boldsymbol{u}^{-1}(x)\cdot{} \boldsymbol{g}_{i,j'}(x)$
is of weight $m$ for $1 \le j' \le n_0 - 1$.
The cost of this attack is $O\left( \binom{m^2}{m}\cdot{}p^2\right)$ which corresponds
to $2^{50.3}$ operations for the specific parameters proposed.

\paragraph{Second Strategy} \label{sec:stg2}

We present another strategy that can be used to recover secret matrices $S$ and 
thus matrices $Q_1,\dots{},Q_{n_0-1}$. This strategy requires to search for codewords of very 
low weight in a linear code. The most efficient algorithm that accomplishes this
task is  the algorithm of \cite{BLP08}  which improves 
upon Stern's algorithm \cite{Stern88}.
However in order to derive a simple bound on the time complexity, we consider this second 
algorithm as in \cite{BaldiChiara07}.
The work factor $\Omega_{n,k,w}$ 
of Stern's algorithm to find $A_w$ codewords of weight 
$w$ in a code of length $n$ and dimension $k$ satisfies 
$\Omega_{k,n,w} \ge \frac{N}{A_w P_w}$
where $(g,\ell)$ are two parameters  and 
$N$ is the number of binary operations required for each iteration 
\begin{equation} \label{eq:iteraStern}
N = (n-k)^3/2 + k(n-k)^2 + 2g \ell \binom{k/2}{g} + 2g(n-k) \frac{\binom{k/2}{g}^2}{2^\ell}.
\end{equation}
$P_w$ represents the probability of finding a given codeword of weight $w$
$$
P_w = \frac{\binom{w}{g} \binom{n-w}{k/2-g}}{\binom{n}{k/2}} \frac{\binom{w-g}{g}\binom{n-k/2-w+g}{k/2-g}}{\binom{n-k/2}{k/2}}
\frac{\binom{n-k-w+2g}{\ell}}{\binom{n-k}{\ell}}.
$$
Recall that $G_{\le k}^{-1}$ is specified by polynomials 
$\boldsymbol{g}_{i,j}(x)$. Let $\boldsymbol{d}_{i,j}(x)$ be the polynomial
$\boldsymbol{g}_{i,j}(x) \cdot{} \boldsymbol{g}_{i,1}^{-1}(x) \mod (x^p-1)$ and consider
the code $\code{E}_{i}$ defined by the following generator matrix:
$$
E_i = \left (
\begin{array}{cccc}
I_p & D_{i,2} & \cdots{} & D_{i,n_0-1}
\end{array}
\right)
$$
where as usual the circulant matrix $D_{i,j}$ is characterised by 
 the polynomial  $\boldsymbol{d}_{i,j}(x)$. Then  $\code{E}_{i}$ contains at least $p$ codewords
of low weight $(n_0-1)m = 21$ since
$$
S_{i,1} \times E_i = 
\left (
\begin{array}{cccc}
S_{i,1} & S_{i,2} & \cdots{} & S_{i,n_0-1}
\end{array}
\right).
$$
It is therefore possible to recover matrices $S_{i,1},\dots{},S_{i,n_0-1}$ 
with a complexity of  $2^{32}$ operations by applying Stern's algorithm
with $(g,\ell) = (3,43)$ in order to find a codeword of weight $21$ in a code of dimension $p$ and length 
$(n_0-1)p = 12096$.

\subsubsection{Extraction of the Secret Code}

After recovering $S$, $Q_1,\dots{},Q_{n_0-1}$, 
one is therefore able to compute the following generator matrix $\tilde{G}$
defined by:
$$
\tilde{G} 
= G' \times 
\left (
\begin{array}{*{4}{c}}
I_p   & 0 &  \cdots{} & 0\\
0  & \ddots{}  & \ddots{} & \vdots{} \\
\vdots{} & \ddots{}  & I_p & 0 \\
0 & \cdots{}   & 0 & Q_{n_0}^{-1}
\end{array}
\right)
= \left (
\begin{array}{ccc|c}
            & & & A_1 \\
& I_{k} &   & \vdots{}\\
&       &     & A_{n_0-1}
\end{array}
\right) 
$$
where for $1\le i \le n_0-1$, we set $A_i =  (H_{n_0}^{-1} \times H_i)^T \times Q_{n_0}^{-1}$. 
Recall that matrices $H_1,\dots{},H_{n_0}$ and $Q_{n_0}$ are still unknown. However,
one can easily check that for any different $i$ and $j$, 
we also have  $(A_i\times A_j^{-1})^{T}$ = $H_i \times H_j^{-1}$ whenever
$H_j$ is invertible. Thus,
if we set $B_{i,j} = (A_i\times A_j^{-1})^T$  then for a fixed 
$1\le i \le n_0-1$ and for any different integers $j$ and $j'$, we
have that $H_j\times B_{i,j} = H_{j'} \times B_{i,j'} = H_i$. 
Consider now the code defined by the following generator matrix $G_1$:
$$
G_1 = \left (
\begin{array}{*{4}{c}}
I_p & B_{2,1} & \cdots{} & B_{n_0-1,1}
\end{array}
\right).
$$
It is easy to see that $H_1 \times G_1 = \left (
\begin{array}{*{4}{c}}
H_1 & H_{2} & \cdots{} & H_{n_0-1}
\end{array}
\right)
$. This also means that $G_1$ spans a code with a minimum distance that is smaller 
than $(n_0-1)d_v$. Therefore, by applying dedicated algorithms 
(\cite{CanteautChabaud98} or \cite[Volume I, Chapter 7]{Handb98}) searching for codewords of small
weight, it is possible to recover matrices $H_1,\dots{},H_{n_0-1}$. For instance,
the work factor of  Stern's algorithm  for searching codewords of weight 
$(n_0-1)d_v = 3*13=39$ in a code of dimension $p = 4032$ and
length $p(n_0-1)= 12096$ is about $2^{37}$ operations with $(g,\ell) = (3,43)$.

Finally, we are able to compute  $(H_i^T)^{-1}\times A_i = (H_{n_0}^{-1})^T \times Q_{n_0}^{-1}$
for any $1\le i \le n_0-1$. Inverting this matrix and applying again the second strategy
presented in Section~\ref{sec:stg2}, it is possible to find the matrices 
$H_{n_0}$ and $Q_{n_0}$.
  
 \subsection{Example}
We illustrate the  previously described attacks with some randomly generated
 polynomials $\boldsymbol{s}_{i,j}(x)$ and $\boldsymbol{q}_{i,j}(x)$ of weight $m=7$ and degree
 less than $p= 4032$ as given in \cite{BaldiChiara07}. 
We only put the exponents of the monomials
 that intervene in the expression of the polynomials. Recall that some coefficients
 $\boldsymbol{s}_{i,j}(x)$ has to be of even weight (actually of weight $m-1=6$) in order
 to generate an invertible matrix $S$.
 We implemented the attack in MAGMA software \cite{MAGMA}. 
 The running time on a Pentium 4 (2.80GHz) with 500 Mbytes RAM for 
 the  second strategy is $140$ seconds. The last step that consists in recovering the secret LDPC code 
is performed by applying Canteaut-Chabaud algorithm. The work factor of this operation
is  about $2^{36}$ operations. Our
implementation in MAGMA software finds a codeword of weight $(n_0-1)d_v =39$ in about $15$ minutes.
{\small \begin{eqnarray*}
 H_{1} &=&  [ 213, 457, 1467, 1702, 1786, 2015, 2155, 2197, 2569, 2744, 2823, 2902, 3710]\\
 H_{2} &=& [ 6, 626, 868, 1102, 1564, 1894, 2401, 2595, 2982, 3570, 3605, 3771, 3835]\\
 H_{3} &=& [ 615, 639, 1198, 1513, 1712, 1850, 1941, 2397, 2553, 3074, 3373, 3798, 3960]\\
H_{4}  &=& [ 135, 149, 241, 735, 1265, 2075, 2869, 3111, 3218, 3625, 3760, 3785, 
3969 ]
 \end{eqnarray*}
 \begin{eqnarray*}
 S_{1,1} &=& [ 24, 274, 334, 2025, 2574, 2661, 3601]\\
 S_{1,2} &=& [ 512, 1177, 2524, 2526, 2904, 2968, 3340]\\
 S_{1,3} &=& [ 930, 1175, 1210, 1459, 2200, 2303, 2811]\\
 S_{2,1} &=& [ 503, 1258, 1632, 1658, 2055, 2221, 2764]\\
 S_{2,2} &=& [ 989, 1256, 2568, 2625, 2906, 3139]\\
 S_{2,3} &=& [ 561, 616, 2499, 2787, 2835, 3061, 3865]\\
 S_{3,1} &=& [ 177, 465, 1659, 1958, 2795, 3605]\\
 S_{3,2} &=& [ 419, 461, 1540, 2262, 2435, 3474, 3587]\\
 S_{3,3} &=& [ 554, 1119, 1307, 2018, 2193, 2631, 3755]
 \end{eqnarray*}
 \begin{eqnarray*}
 Q_{1} &=& [ 456, 578, 1551, 1562, 1992, 2919, 3476]\\
 Q_{2} &=& [ 250, 268, 897, 1782, 2127, 3163, 3378]\\
 Q_{3} &=& [ 14, 1132, 1672, 1716, 2164, 2723, 3409]\\
 Q_{4} &=& [ 443, 593, 2401, 2615, 2981, 3612, 3993]
 \end{eqnarray*}
}

\section{Conclusion} \label{sec:conclusion}

The idea to introduce quasi-cyclic codes and quasi-cyclic 
low density parity-check codes is motivated by practical concerns to 
reduce key sizes of McEliece cryptosystem. The first variant of \cite{Gab05}
uses quasi-cyclic codes obtained from subcodes of a cyclic BCH code. The other variant of \cite{BaldiChiara07}
uses quasi-cyclic low density parity-check codes.
However, we have shown here 
that the cost of these two attempts at reducing key size is made at the expense of the security. Indeed,
we have presented different structural cryptanalysis of these two variants of McEliece cryptosystem. 
The first attack is applied to the variant of \cite{Gab05}
and extracts the secret permutation supposed to hide the structure 
of the secret codes. We show that the secret key recovery amounts to 
solve an over-constrained linear system. 
The  second attack accomplishes a total break of \cite{BaldiChiara07}. 
In the first phase, we
look for divisors of low weight of a given public polynomial. 
The last phase recovers the secret parity check matrix of 
the secret quasi-cyclic low density parity-check code by
looking for low weight codewords in a punctured version of the secret code.
An implementation shows that the first phase can be accomplished 
in about 140 seconds and the second phase in about 15 minutes.

However these results cannot be applied to the original 
McEliece's scheme using Goppa codes which 
represents up to now the only unbroken scheme.  An open problem 
which would be desirable to solve  is to come up with a way of reducing significantly the key
sizes in this type of public-key cryptosystem by maintaining the security intact.

\bibliographystyle{plain}

\begin{thebibliography}{10}

\bibitem{BaldiChiara07}
M.~Baldi and G.~F. Chiaraluce.
\newblock Cryptanalysis of a new instance of {McEliece} cryptosystem based on
  {QC-LDPC} codes.
\newblock In {\em IEEE International Symposium on Information Theory}, pages
  2591--2595, Nice, France, March 2007.

\bibitem{BMT78}
E.~R. Berlekamp, R.~J. McEliece, and H.~C.~A. van Tilborg.
\newblock On the intractability of certain coding problems.
\newblock {\em IEEE Transactions on Information Theory}, 24(3):384--386, 1978.


\bibitem{BLP08}
D.J. Berstein, T.~Lange, and C.~Peters.
\newblock Attacking and defending the {McE}liece cryptosystem.
\newblock PQCrypto, pages 31--46, 2008.

\bibitem{MAGMA}
W.~Bosma, J.~J. Cannon, and C.~Playoust.
\newblock The {Magma} algebra system {I}: The user language.
\newblock {\em J. Symb. Comput.}, 24(3/4):235--265, 1997.

\bibitem{MonRosShok00}
A.~Shokrollahi C.~Monico, J.~Rosenthal.
\newblock Using low density parity check codes in the {McEliece} cryptosystem.
\newblock In {\em IEEE International Symposium on Information Theory (ISIT
  2000)}, page 215, Sorrento, Italy, 2000.

\bibitem{CC94}
A.~Canteaut and H.~Chabanne.
\newblock A further improvement of the work factor in an attempt at breaking
  {McE}liece's cryptosystem.
\newblock In {\em EUROCODE 94}, pages 169--173. INRIA, 1994.

\bibitem{CC95}
A.~Canteaut and F.~Chabaud.
\newblock Improvements of the attacks on cryptosystems based on
  error-correcting codes.
\newblock Technical Report 95--21, INRIA, 1995.

\bibitem{CanteautChabaud98}
A.~Canteaut and F.~Chabaud.
\newblock A new algorithm for finding minimum-weight words in a linear code:
  Application to {McEliece}'s cryptosystem and to narrow-sense {BCH} codes of
  length 511.
\newblock {\em IEEE Transactions on Information Theory}, 44(1):367--378, 1998.

\bibitem{CayOtmVer07}
P.L. Cayrel, A.~Otmani, and D.~Vergnaud.
\newblock {On Kabatianskii-Krouk-Smeets Signatures}.
\newblock In {\em {Proceedings of the first International Workshop on the
  Arithmetic of Finite Fields (WAIFI 2007)}}, {Springer Verlag Lecture Notes},
  pages 237--251, {Madrid, Spain}, June~21--22 2007.

\bibitem{EOS06}
D.~Engelbert, R.~Overbeck, and A.~Schmidt.
\newblock A summary of {McEliece}-type cryptosystems and their security.
\newblock In {\em Journal of Mathematical Cryptology}, volume~1, pages
  151--199, 2007.

\bibitem{Gab05}
P.~Gaborit.
\newblock Shorter keys for code based cryptography.
\newblock In {\em Proceedings of the 2005 International Workshop on Coding and
  Cryptography ({WCC} 2005)}, pages 81--91, Bergen, Norway, March 2005.

\bibitem{GabGir07}
P.~Gaborit and M.~Girault.
\newblock Lightweight code-based authentication and signature.
\newblock In {\em IEEE International Symposium on Information Theory (ISIT
  2007)}, pages 191--195, Nice, France, March 2007.

\bibitem{GabLaurSend07}
P.~Gaborit, C.~Lauradoux, and N.~Sendrier.
\newblock Synd: a fast code-based stream cipher with a security reduction.
\newblock In {\em IEEE International Symposium on Information Theory (ISIT
  2007)}, pages 186--190, Nice, France, March 2007.

\bibitem{LeeBrick88}
P.~J. Lee and E.~F. Brickell.
\newblock An observation on the security of {McEliece}'s public-key
  cryptosystem.
\newblock In {\em Advances in Cryptology - EUROCRYPT'88}, volume 330/1988 of
  {\em Lecture Notes in Computer Science}, pages 275--280. Springer, 1988.

\bibitem{Leon88}
J.~S. Leon.
\newblock A probabilistic algorithm for computing minimum weights of large
  error-correcting codes.
\newblock {\em IEEE Transactions on Information Theory}, 34(5):1354--1359,
  1988.

\bibitem{LiDenWan94}
Y.~X. Li, R.~H. Deng, and X.-M. Wang.
\newblock On the equivalence of {McEliece's} and {Niederreiter's} public-key
  cryptosystems.
\newblock {\em IEEE Transactions on Information Theory}, 40(1):271--273, 1994.

\bibitem{McEliece78}
R.~J. McEliece.
\newblock {\em A Public-Key System Based on Algebraic Coding Theory}, pages
  114--116.
\newblock Jet Propulsion Lab, 1978.
\newblock DSN Progress Report 44.

\bibitem{MinderShokrollahi07}
L.~Minder and A.~Shokrollahi.
\newblock Cryptanalysis of the {S}idelnikov cryptosystem.
\newblock In {\em Eurocrypt 2007}, volume 4515 of {\em Lecture Notes in
  Computer Science}, pages 347--360, Barcelona, Spain, 2007.

\bibitem{Niederreiter86}
H.~Niederreiter.
\newblock Knapsack-type cryptosystems and algebraic coding theory.
\newblock {\em Problems Control Inform. Theory}, 15(2):159--166, 1986.

\bibitem{Handb98}
V.S. Pless and W.C. Huffman, editors.
\newblock {\em Handbook of coding theory}.
\newblock North Holland, 1998.

\bibitem{Sidelnikov94}
V.M. Sidelnikov.
\newblock A public-key cryptosystem based on binary {Reed}-{Muller} codes.
\newblock {\em Discrete Mathematics and Applications}, 4(3), 1994.

\bibitem{SidelShesta92}
V.M. Sidelnikov and S.O. Shestakov.
\newblock On the insecurity of cryptosystems based on generalized
  {Reed-Solomon} codes.
\newblock {\em Discrete Mathematics and Applications}, 1(4):439--444, 1992.

\bibitem{Stern88}
J.~Stern.
\newblock A method for finding codewords of small weight.
\newblock In G.~D. Cohen and J.~Wolfmann, editors, {\em Coding Theory and
  Applications}, volume 388 of {\em Lecture Notes in Computer Science}, pages
  106--113. Springer, 1988.

\bibitem{Gabidulin95}
E. M. Gabidulin.
\newblock Public-Key Cryptosystems Based on Linear Codes.
\newblock 1995.

\bibitem{BergerLoidreau05}
T.P. Berger and P. Loidreau.
\newblock How to Mask the Structure of Codes for a Cryptographic Use.
\newblock Des. Codes Cryptography, 35(1):63--79, 2005.

\bibitem{Wieschbrink06}
C. Wieschebrink.
\newblock An Attack on a Modified Niederreiter Encryption Scheme.
\newblock Public Key Cryptography - PKC 2006, Volume 3958/2006 of {\em Lecture Notes in Computer Science}, pages
  14--26. Springer, 2006.

\bibitem{Sendrier98} 
N. Sendrier. 
\newblock On the Concatenated Structure of a Linear Code.
\newblock AAECC, 9(3):221--242, November 1998.

\end{thebibliography}

\end{document}